\documentclass[review]{elsarticle}

\usepackage{amssymb,amsmath}
\usepackage{amsthm}
\usepackage{graphics}
\usepackage{subfigure}
\usepackage{hyperref}
\usepackage{cases}
\usepackage{bm}

\usepackage{algorithm}
\usepackage[noend]{algpseudocode}

\newtheorem{theorem}{Theorem}[section]

\newtheoremstyle{prestyle}
{0} 
{\topsep} 
{\itshape} 
{} 
{\bfseries} 
{.} 
{.5em} 
{} 
\theoremstyle{prestyle}	
\newtheorem{lemma}[theorem]{Lemma}

\theoremstyle{definition}	
\newtheorem{example}{Example}

\journal{Information Sciences}

\makeatletter
\def\BState{\State\hskip-\ALG@thistlm}
\makeatother

\begin{document}

\begin{frontmatter}

\title{Efficient methods to determine the reversibility of general 1D linear cellular automata in polynomial complexity}

\author[a]{Xinyu Du}
\ead{dothingyo@gmail.com}
\author[a]{Chao Wang\corref{cor1}}
\ead{wangchao@nankai.edu.cn}
\author[b]{Tianze Wang}
\ead{blakewang@live.com}
\author[a]{Zeyu Gao}
\ead{gaoshangyouxiang@foxmail.com}

\tnotetext[mytitlenote]{Fully documented templates are available in the elsarticle package on 
\href{http://www.ctan.org/tex-archive/macros/latex/contrib/elsarticle}{CTAN}.}




\address[mymainaddress]{Address: College of Software, Nankai University, Tianjin 300350, China}
\address[secondaddress]{Address: Department of Software and Computer Systems, KTH Royal Institute of Technology, Stockholm 16440, Sweden}

\cortext[cor1]{Corresponding author.}


\begin{abstract}
In this paper, we study reversibility of one-dimensional(1D) linear cellular automata(LCA) under null boundary condition, 
whose core problems have been divided into two main parts: calculating the period of reversibility and 
verifying the reversibility in a period. 
With existing methods, the time and space complexity of these two parts are still too expensive to be employed. 
So the process soon becomes totally incalculable with a slightly big size, which greatly limits its application. 
In this paper, we set out to solve these two problems using two efficient algorithms, 
which make it possible to solve reversible LCA of very large size. 
Furthermore, we provide an interesting perspective to conversely generate 1D LCA from a given period of reversibility. 
Due to our methods' efficiency, we can calculate the reversible LCA with large size, 
which has much potential to enhance security in cryptography system.
\end{abstract}


\end{frontmatter}


\newpage
\section{Introduction \label{Introduction}}
Cellular automata (CA) are discrete dynamical systems and models of massively parallel computation that share many properties of the physical world. 
It's one of the most classic models ever proposed, and it has been widely applied to parallel systems\cite{das2010parallel}, 
secret sharing\cite{maranon2005new}, image encryption\cite{souyah2016fast}, traffic simulation\cite{rickert1996two}, 
thermodynamic simulation\cite{takesue1989ergodic},skin disease diagnosis\cite{kippenberger2013modeling} and many other fields.
Recently, Stephen Wolfram summarized classic CA papers, which has aroused resurgent research interests in CA\cite{wolfram2018cellular}.

As for reversibility, a reversible CA is a CA in which every configuration has a unique predecessor. 
So the previous state of any cell before an update can be determined uniquely from the updated states of all the cells. 
Reversibility of CA has broadened its applications in image security and have a fairly good performance. 
There are a lot of studies about it. \cite{kari1992cryptosystems,wang2013novel,Zhang2002}

Thanks to CA's property of finiteness, its studies mainly fall into 3 categories: periodic boundary, reflective boundary and null boundary. 
Before detailing the null boundary condition in this paper, we want to introduce some good studies of the other two categories as well.

For reflective boundary, LCA of radius 1 has been well studied\cite{2012_H_Akin}. 
For periodic boundary, there are some good studies of the following cases: the hybrid elementary CA(ECA) 90/150\cite{1998_P_Sarkar}, 
the ECA 150\cite{2007_L_Hernandez_Encinas, 2013_A_Martin_del_Rey_b}, the five-neighbor case $\mathbb{F}_p$\cite{2011_Z_Cinkir,2012_I_Siap}, 
the general 1D CA case\cite{2004_A_Nobe}, CA with memory\cite{2012_J_C_Seck-Tuoh-Mora},
$\sigma/\sigma^+$-automata\cite{2000_K_Sutner}, undecidable 2D case\cite{1990_J_Kari} and hexagonal 2D CA\cite{2013_S_Uguz}. 
There are also novel methods like Welch Sets\cite{seck2017welch}, Cayley Trees\cite{chang2017reversibility_1} and graphs\cite{seck2018graphs}.

As for null boundary condition, which is discussed in this paper,
studies fall into one-dimensional and multidimensional cases.
In the one-dimensional case, the following rules have been studied: the hybrid ECA 90/150\cite{1998_P_Sarkar}, 
the ECA 150\cite{2006_A_M_Del_Rey}, the symmetric no-dummy-neighbor linear rules\cite{2009_A_Martin_Del_Rey}, 
the five-neighbor linear rule 11111\cite{2011_A_Martin_del_Rey} and the ECA 150 over finite state set $\mathbb{F}_p$\cite{2013_A_Martin_del_Rey_a}. 
As for multidimensional cases, there are results about the 2D linear rules group 2460 over the ternary field $\mathbb{Z}_3$\cite{2010_I_Siap}, 
hexagonal 2D CA \cite{2011_I_Siap} and 3D linear rules over $\mathbb{Z}_m$ \cite{2012csah}. 
Recently, Chang\cite{chang2017reversibility_2} proposed a criterion for testing the reversibility of a 
multidimensional linear cellular automaton under null boundary condition and an algorithm for the computation of its reverse.

There are two existing methods to deal with the reversibility of 1D LCA over binary field $\mathbb{Z}^{2}$ under null-boundary condition. 
The first is to use transition matrix \cite{2011_A_Martin_del_Rey} and the second is to construct a DFA \cite{yang2015reversibility}. 
We will analyze these two methods and compared our algorithms with them later.

The aim of this paper is to design and optimize algorithms to efficiently calculate reversibility of 1D LCA, which
make reversibility of bigger size calculable and thus greatly extend reversibility of LCA's area to use. 
In practice, one LCA corresponds with one polynomial. 
For any given linear rule $\lambda_{1}...\lambda_{m}$($\lambda_{1}$ = $\lambda_{m} = 1$), 
we can find its corresponding polynomial $f(x) = 1 + \sum_{i=1}^{m-1}\lambda_{m-i}x^{i}$.
For example, 1011011 corresponds to $x^{6}+ x^{4} + x^{3} + x + 1$. 
If we have a rule whose $\lambda_{1} = 0$ or $\lambda_{m} = 0$, 
then for the calculation of period we can simplify it to a smaller rule 
whose $\lambda_{1}$ = $\lambda_{m} = 1$ by eliminating ``0" in both borders and give a corresponding polynomial 
after that. Moreover, we prove that the period of reversibility of LCA equals to the polynomial period in Section \ref{Period}. 
Based on this period, we just need to verify the reversibility within it.

The paper is organized as follows:
\textbf{Section \ref{Mathematical description}} describes the basic mathematical symbols used in this paper. 
\textbf{Section \ref{Existing}} gives a brief introduction of two existing methods. 
\textbf{Section \ref{Period}} details using polynomial method to calculate the CA period. 
\textbf{Section \ref{Reversibility}} optimizes an algorithm to give reversibility in a period. 
\textbf{Section \ref{Generate}} conversely generates LCA from a given period. 
\textbf{Section \ref{Conclusion}} makes a brief summary.


\section{Mathematical description\label{Mathematical description}}
\subsection{Basic definitions}
$p$ is a prime number which represents the number of elements in a finite field. In this paper we always suppose $p=2$.\\
\\
Four-tuple $A=\{d,S,\vec{N},f\}$ is usually used to describe CA.
\begin{itemize}
	\item
	$d\in\mathbb{Z}_{+}$ denotes the dimension of the cellular space.
	\item
	$S = \{0,1,...,p-1\}$ is a finite state set including all states of any cell at any time.
	\item
	$\vec{N}=(\vec{n_1},\vec{n_2},\ldots,\vec{n_m})$ is the neighbor vector, 
     in which $\vec{n_i}\in\mathbb{Z}^d$ and $\vec{n_i}\neq\vec{n_j}$ when $i\neq j$. 
     Therefore, the neighbors of the cell $\vec{n} \in \mathbb{Z}_{d}$ are the $m$ cells $\vec{n}+\vec{n_{i}}$, $i=1,2,\ldots,m$. 
     Here $m$ is called the size of the LCA.
	\item
	$f:S^m\to S$ is a local rule, which maps the current states of all neighbors of a cell to the next state of this cell.
\end{itemize}

A configuration is a mapping $c$: $\mathbb{Z}^{d} \to S$ which assigns each cell a state. 
If we use $c^t$ to denote the configuration(states of all cells) at time $t$, 
then the state of cell $\vec{n}$ at time $t$ is $c^t(\vec{n})$ and state at time $t+1$ goes like
\[
c^{t+1}(\vec{n}) = f( c^t(\vec{n}+\vec{n_1}), c^t(\vec{n}+\vec{n_2}), \ldots, c^t(\vec{n}+\vec{n_m}) ).
\]

For linear cellular automata(LCA), the local rule $f$ should be a linear function as follows:
\begin{equation}
\begin{aligned}
&f( c(\vec{n}+\vec{n_1}), c(\vec{n}+\vec{n_2}), \ldots, c(\vec{n}+\vec{n_m}))\\
=& [\lambda_1c(\vec{n}+\vec{n_1}) + \lambda_2c(\vec{n}+\vec{n_2})+ \cdots +
\lambda_mc(\vec{n}+\vec{n_m})] \ \bmod \ p,\label{eq:first}
\end{aligned}
\end{equation}
where $\lambda_i\in S$ is the rule coefficient of cell $\vec{n}+\vec{n_i}$, $i$ = 1,2,...,$m$.		

For an LCA, the next state of a cell is the linear combination of the rules coefficients and its neighbors' current states modulo $p$. 
Usually, there are infinite number of cells in a $d$-dimensional integer space $\mathbb{Z}^{d}$. 
However, when we try to apply CA to practical problems, the cellular space must be finite. 
So boundary conditions are taken into consideration and null boundary is most often used.
In the condition of a null boundary:
$c(\vec{n}) \equiv 0$ for $\vec{n} \in \mathbb{Z}^{d}$ and $\vec{n}$ not in the cellular space.

In this paper, we discuss 1D null-boundary LCA of $n$ cells over $\mathbb{Z}^{2}$, 
in which $p = 2$, $d$ = 1, $S$ = $\mathbb{Z}_{2}$ = $\{0, 1\}$.

We use $s^{t}_{i}$ to denote the state of cell $i$ at time $t$, in which $i = 1,2,...,n$.
Then we define neighbor vector as ($-r_{L}$,...,0,...,$r_{R}$), 
linear coefficients as $\lambda_{-r_{L}}$,...,$\lambda_{0}$,...,$\lambda_{r_{R}}$.
Then according to \textbf{Eq.(\ref{eq:first})}, the next state of cell $i$ is defined as follows:

\begin{equation}
\begin{aligned}
s_{i}^{t+1} = [\lambda_{-r_{L}}s_{i-r_{L}}^{t} +...+ \lambda_{-1}s_{i-1}^{t} + \lambda_{0}s_{i}^{t} + 
\lambda_{1}s_{i+1}^{t} +...+ \lambda_{r_{R}}s_{i+r_{R}}^{t}]\quad \mod \quad 2,\label{eq:second}
\end{aligned}
\end{equation}
in which $\lambda_{-r_L}=1$ when $r_L>0$ and $\lambda_{r_R}=1$ when $r_R>0$.
If we are given a rule whose $\lambda_{-r_L} = 0$ or $\lambda_{r_R} = 0$, 
we delete the border elements until it satisfies $\lambda_{-r_L}=1$ and $\lambda_{r_R}=1$.
Moreover, under null boundary condition, $s_i^t=0$ for $i\notin\{1,2,\ldots,n\}$.

From \textbf{Eq.(\ref{eq:second})}, we know that a linear rule can be denoted by LCA's linear coefficients and its neighbor vector.
We use $\lambda_{-r_L}...\lambda_{0}...\lambda_{r_R}$ and $(-r_L,\ldots,0,\ldots,r_R)$ to represent linear coefficients and neighbor vector.

As examples, in \cite{2011_A_Martin_del_Rey}, the linear rule 11111 with neighbor vector(-2, -1, 0, 1, 2) were discussed.
The rule $\lambda_{i}=1$, $(i=-k,\dots,k)$ with neighbor vector $(-k,\ldots,-1,0,1,\ldots,k)$ were discussed in \cite{2009_A_Martin_Del_Rey}.

Moreover, the condition of unilateral rules$(r_L=0 \mbox{ or } r_R=0)$ has been discussed \cite{yang2015reversibility} 
and they got the following conclusion:
a one-dimensional LCA with a unilateral rule under null boundary conditions is always irreversible if $\lambda_{0} = 0$ 
and always reversible if $\lambda_{0} = 1$.
So in this paper, we just discuss the bilateral condition $r_L>0$ and $r_R>0$.

\section{Two existing solutions\label{Existing}}

\subsection{\textbf{Solution1: Transition matrix}}
In 2011, Mart\'in del Rey and Rodr\'iguez S\'anchez explained the reversibility of LCA with analysis of transition matrix 
and achieved great performance for case 11111\cite{2011_A_Martin_del_Rey}. 
In this case, for a transition matrix $M_{n}$ over $\mathbb{Z}_{2}$, the determinant $|M_{n}|$ can only be 0 or 1. 
When $|M_{n}|=1$, it is invertible, vice versa. 
They analyzed the determinant of the transition matrix and demonstrated that the determinant 
of the $\{n+5\}^{th}$ matrix identically equals to the determinant of ${n}^{th}$ matrix. 
Furthermore, they proved that when $n=5k$ or $5k+1$, it is invertible, in which $n$ is the cell number of the LCA.

\subsubsection{Brief description}
Here is the simple process of the linear rule 11111 using the first solution. For further details, please refer to \cite{2011_A_Martin_del_Rey}.

\begin{enumerate}
\item Configuration transition:
$ M_{n} * (c^{t})^{T} = (c^{t+1})^{T}$,
in which $M_{n}$ denotes the transition matrix, $c^{t}$ and $c^{t+1}$ denote the configurations at $t$ and $t+1$.

\item Transition matrix of case 11111:
$$
M_{n}={
	\left[ \begin{array}{cccccccc}
	1 & 1 & 1 & 0 & \cdots & 0 & 0 & 0\\
	1 & 1 & 1 & 1 & \cdots & 0 & 0 & 0\\
	1 & 1 & 1 & 1 & \cdots & 0 & 0 & 0\\
	0 & 1 & 1 & 1 & \cdots & 0 & 0 & 0\\
	\vdots & \vdots & \vdots & \vdots & \ddots & \vdots & \vdots & \vdots\\
	0 & 0 & 0 & 0 & \cdots & 1 & 1 & 1 \\
	0 & 0 & 0 & 0 & \cdots & 1 & 1 & 1 \\
	0 & 0 & 0 & 0 & \cdots & 1 & 1 & 1 \\
	\end{array}
	\right ]}.
$$

\item Conclusion:


The matrix can be written as follows:
$$
M_{n}={
	\left[ \begin{array}{cc}
	M_{5} & A \\
	A^{T} & M_{n-5} \\
	\end{array}
	\right ]}
$$

and then it can be proved 
$|M_{n}| = |M_{n-5}|$ $\quad \mod \quad 2$\\

Furthermore, the determinant of $M_{n}$ satisfies:

$$
|M_{n}| \quad \mod \quad 2 =
\begin{cases}
1,\quad if \quad n = 5k \quad or \quad n = 5k + 1,\quad with \quad k \ge 0,\\
0,\quad otherwise.	
\end{cases}
$$\\

To compute the inverse matrix $M_{n}^{-1}$, they use the following decomposition:

If $n$ = $5k$, then
$$
{M_{n}}^{-1}={
	\left[ \begin{array}{cccc}
	M_{5}^{-1} & B & \cdots & B \\
	B^{T} & M_{5}^{-1} & \ddots & \vdots\\
	\vdots & \ddots & \ddots & B \\
	B^{T} & \cdots & B^{T} & M_{5}^{-1}\\
	\end{array}
	\right ]}
$$

If $n$ = $5k+1$, then
$$
{M_{n}}^{-1}={
	\left[ \begin{array}{cccc}
	M_{6}^{-1} & Q & \cdots & Q \\
	Q^{T} & R & \ddots & \vdots\\
	\vdots & \ddots & \ddots & Q \\
	Q^{T} & \cdots & Q^{T} & R\\
	\end{array}
	\right ]}
$$

$B,R,Q$ are three constant matrixes.

\end{enumerate}

\subsubsection{Basic performance analysis}

There are two choices using this method to find period of reversibility: the first is to calculate determinant,
find and prove period manually, which has great performance for specific case but is difficult to generalize. 
The second is to find period with determinant value sequence of 0 and 1, which saves human labor, 
but expensive in calculation cost and difficult to prove periodicity in mathematics.

For generalization, we use the second choice to calculate period through transition matrix, 
which also needs to cost $O(n^3)$ for increasing cell number $n$ by row reducing the transition matrix to an echelon form 
to calculate determinant. Furthermore, for verifying a period, we at least cost $O(\sum_{n = 1}^{period} n^{3}) \approx O(period^{4})$. In the worst case, $period = 2^{r_{R}+r_{L}} - 1$ and thus cost = $O({2^{4(r_{R}+r_{L})}})$,
which is pretty expensive.

\subsection{\textbf{Solution2: DFA}}
\label{DFA}
Inspired by the former result\cite{2011_A_Martin_del_Rey}, Bin Yang and Chao Wang proposed DFA to solve the problem 
of the reversibility of general 1D LCA over the binary field $\mathbb Z_{2}$ under null boundary condition\cite{yang2015reversibility}.
First they constructed a DFA with all $2^{r_{R}}$ tuples in a node.
Then they got the period by judging if nodes repeated. Finally they verified reversibility by judging if postfixes repeat in a node.

\begin{figure}[!htbp]
	\small
	\centering
	\includegraphics[width=18cm]{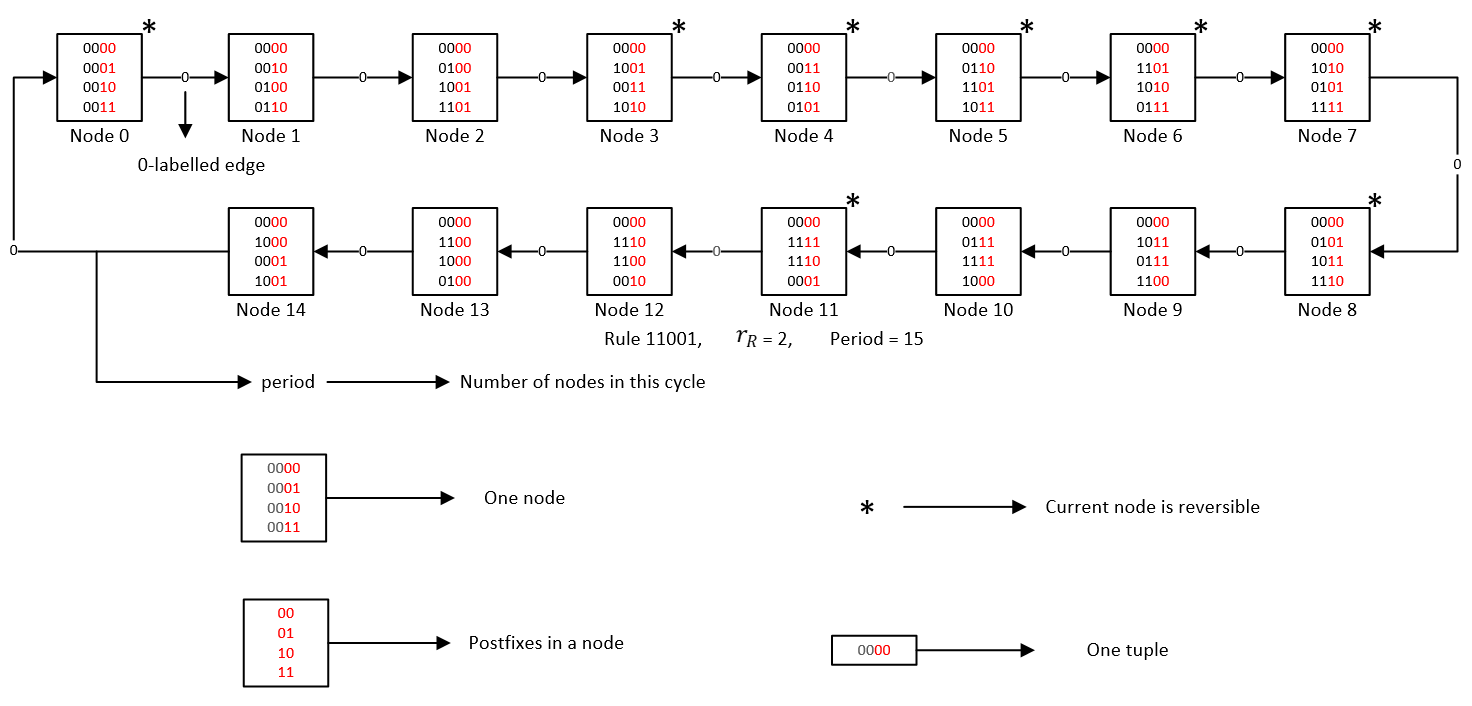}
	\caption{The DFA of linear rule 11001 with the neighbor vector (-2, -1, 0, 1, 2)}\label{flow1}
\end{figure}

\textbf{Fig.\ref{flow1}} shows the former DFA with tuples in each node. In these DFA figures, 
the tuple $(0,0,0,1)$ is simplified to $0001$. In the theoretical demonstration part, we use $(0,0,0,1)$ to denote the tuple.

\newpage

\subsubsection{Basic definitions of DFA}
we use the following terms to describe DFA. For further details, please refer to \cite{yang2015reversibility}.
\begin{itemize}
	\item
	node: $2^{r_{R}}$ tuples contained in one rectangle in the DFA figure.
	\item
	tuple: one row of a node, length of a tuple equals to $r_{L}+r_{R}$.
	\item
	period: the number of nodes in a DFA figure.
	\item
	prefix: the first $r_{L}$ elements of a tuple.
	\item
	postfix: the last $r_{R}$ elements of a tuple,\\
	$postfix_{i}$: the $i^{th}$ postfix in a node and
	$postfix_{i} = (e_{1},e_{2},...,e_{i},...,e_{r_{R}})$.
	\item
	symbol *: the current node is reversible (node $i$ is reversible means: LCA with cell number $n = i + k * period $ is reversible, where $k$ is a non-negative integer).
\end{itemize}

\subsubsection{Brief description}
\label{Brief_Description_of_DFA}
Here is the simple process with DFA:

\begin{enumerate}
\label{thm_inheri}
\item Constructing DFA: \\
\begin{description}
	\item[Step1.] First construct the initial node: Set all prefixes as $"0...0"$(all-zeros) of length $r_{L}$, all postfixes as $2^{r_{R}}$ "0-1" permutation of length $r_{R}$. 

	\item[Step2.] Construct 0-labelled edge to the next node.\\	
	\begin{figure}[!htbp]
		\small
		\centering
		\includegraphics[width=18cm]{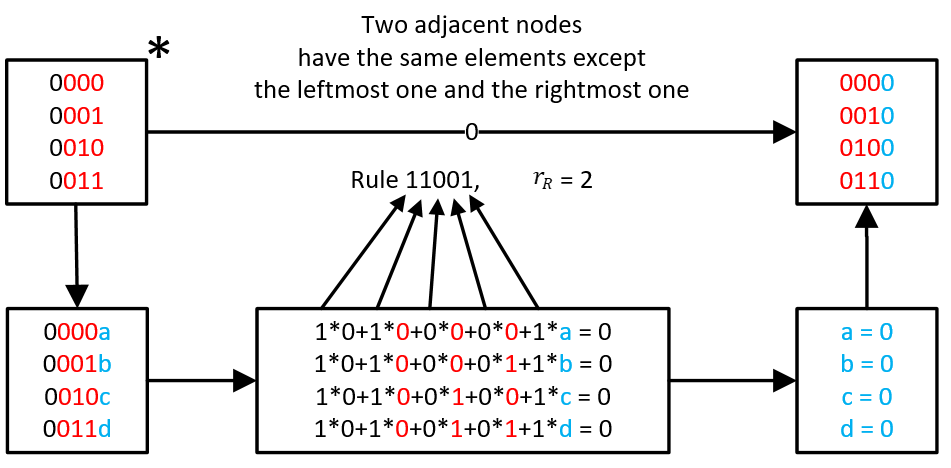}
		\caption{Constructing 0-labelled edge to get the next node}\label{flow_additional}
	\end{figure}
		\label{labelled_edge}\\
		0-labelled edge is the edge between two nodes. With the given "0", we can calculate the next node as follows. For a 0-labelled edge in \textbf{Fig. \ref{flow_additional}}, we need to calculate the next node by solving \textbf{Eq.(\ref{eq:second})}, in which the number "0" on the edge means $s_{i}^{t+1} = 0$, for any given tuple:\\
		\begin{center} 
		$s_{i-r_{L}}^{t} = tuple[0]$\\
		...\\
		$s_{i-1}^{t} = tuple[r_{L}-1]$\\
		$s_{i}^{t} = tuple[r_{L}]$\\
		$s_{i+1}^{t} = tuple[r_{L}+1]$\\
		...\\
		$s_{i+r_{R}-1}^{t} = tuple[r_{L}+r_{R}]$,\\
		\end{center}
		and $s_{i+r_{R}}^{t}$ is the element we want to get, which is also the rightmost element in the next node. Two adjacent tuples have the same elements except the leftmost one and the rightmost one. So after getting the rightmost element of each tuple in the next node, we get the whole next node.
		For example, with linear rule 11001 and neighbor vector (-2, -1, 0, 1, 2), the process can be described by \textbf{Fig.\ref{flow_additional}}.

	\item[Step3.] Repeat \textbf{Step2} until the next node is the same as the initial node to form a 0-edged circle. 

	\item[Step4.] Get the period from the number of nodes in a circle.

	\item[Step5.] Pick up reversible nodes $n_{1},n_{2},\ldots$ whose postfixes do not repeat
and mark these nodes with "*".

	\item[Step6.] Finally, for nodes in the circle with period $p$, the LCA is reversible if and only if $n=n_1,n_2,\ldots \mod p$, where $n$ is the number of cells.

\end{description}

\item Important theorems:\\
	(1) Two adjacent nodes have the same elements except the leftmost element of the first node and the rightmost element of the second node.
		
	(2) The DFA graph is periodic and each node contains $2^{r_R}$ tuples of length $r_L+r_R$. The last edge will lead back to the initial node.

\item Conclusion:\\
	(1) If a 0-edged circle is formed, then the number of nodes in this circle equals to the period.
	
	(2) One node is reversible if and only if all postfixes in it perform a complete permutation of $r_{R}$ binary elements,
which equals to $2^{r_R}$ kinds of all different 0-1 postfixes of length $r_R$ .	
\end{enumerate}

\subsubsection{Brief performance analysis}

The DFA method is applicable to some more general case, but it still needs $O(period*2^{r_R}*r_R)$ space complexity
to store all nodes in a period and $O(period*2^{r_R}*r_R)$ time complexity to judge if they repeat,
both of which are at least exponential complexity. In this way, with $r_{R}$ increasing,
both the period and reversibility become incalculable.

\section{Polynomial method to calculate the period of the reversibility of LCA}
\label{Period}
In this section we discuss one of the core problems: calculating the period of the reversibility of 1D LCA of $N$ cells.

\begin{figure}[!htbp]
	\small
	\centering
	\includegraphics[width=18cm]{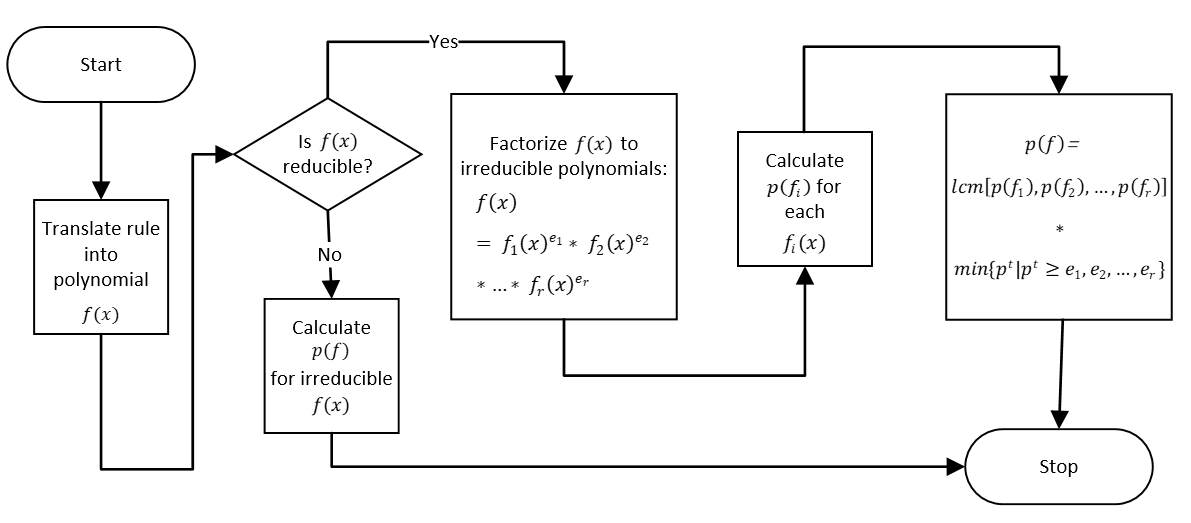}
	\caption{Translate linear rule to polynomial to calculate the period of the reversibility of LCA }
    \label{flow_polynomial}
\end{figure}

\textbf{Fig.\ref{flow_polynomial}} shows the process of calculating the period of the reversibility of LCA by polynomial method.

\begin{theorem}
\label{thm_equal}
The period of the reversibility of LCA is equal to the period of its corresponding DFA.
It is also equal to the period of its corresponding linear feedback shift register(LFSR)
and the period of its corresponding polynomial as well.
\end{theorem}

\begin{proof}
In \cite{yang2015reversibility}, it has been proved that the period of the reversibility of LCA is equal to
the period of its corresponding DFA and the period of its corresponding LFSR.
Furthermore, it is widely known that the period of the LFSR is equal to the period of its
corresponding polynomial from the coding theory.
\end{proof}

\subsection{Calculate the period of irreducible polynomials}
\label{cal_irre}

we assume $f(x)$ is an irreducible polynomial of degree $n$ on a general finite field $GF(q)$, where $q$ is the power of a prime.
Here we demand $f(0) \neq 0$ .

1. Factorize $q^n-1$ into the product of powers of different prime numbers, i.e.,
$q^{n}-1 = p_{1}^{e_{1}} p_{2}^{e_{2}}...p_{r}^{e_{r}}$. Suppose there are $r$ different primes.

2. For $i$ = 1,2,...,$r$, calculate
$(x^{(q^{n}-1)/p_{i}})_{f(x)}, (x^{(q^{n}-1)/p_{i}^{2}})_{f(x)}, (x^{(q^{n}-1)/p_{i}^{3}})_{f(x)}$,...
until we get a nonnegative $f_{i} \leq e_{i}$, such that
$(x^{(q^{n}-1)/p_{i}^{(f_{i})}})_{f(x)} = 1$,
$(x^{(q^{n}-1)/p_{i}^{(f_{i}+1)}})_{f(x)} \neq 1$(the subscript $f(x)$ means mod $f(x)$ in finite field). Thus,
$p(f)|(q^{n} - 1)/p_{i}^{f_{i}}$, $p(f)\nmid(q^{n} - 1)/p_{i}^{f_{i}+1}$.
Then, $p(f)$ = $p^{e_{1}-f_{1}}_{1}$$p^{e_{2}-f_{2}}_{2}$...$p^{e_{r}-f_{r}}_{r}$

As for the details of calculating $(x^{q_{n}-1/p_{i}^{s}})_{f(x)}$, $s=1,2,...$, $i=1,2,\ldots,r$,
please refer to page 150 of ``Algebraic Coding Theory"\cite{berlekamp2015algebraic}.

\subsection{Calculate the period of reducible polynomials for general case}
1. Factorization. There are several good factorization methods we need to mention:
Berlekamp's algorithm \cite{berlekamp1967factoring},
Cantor¨CZassenhaus algorithm \cite{cantor1981new},
Victor Shoup's algorithm \cite{shoup1990new},
we can also see the comparison among these algorithms' performances \cite{von2001factoring}.

In this paper, we use Berlekamp's algorithm to factorize $f(x)$ into the product of irreducible polynomials $f_i(x)$.
This algorithm has a polynomial complexity in fixed finite field and thus is efficient enough over ${\Bbb Z}_2$.

2. Calculate the period of each irreducible polynomial $f_{i}(x)$ by \textbf{Subsection \ref{thm_equal}}.

3. Calculate the period of $f(x)$ from \textbf{Theorem \ref{thm_fac}} after all periods of irreducible polynomials have been obtained.
\begin{theorem}
\label{thm_fac}
Let $f(x)$ be a reducible polynomial and $f(0)\neq 0$, Assume
$f(x)=f_1(x)^{e_1}f_2(x)^{e_2}\cdots f_r(x)^{e_r}$, where
$f_1(x)$, $f_2(x)$, $\ldots$, $f_r(x)$ are $r$ different irreducible polynomials over ${\Bbb Z}_2$
and $e_1$, $e_2$, $\ldots$, $e_r$ are $r$ positive integers, then we get:
\begin{equation}
p(f)=lcm\left[ p(f_1),p(f_2),...,p(f_r)\right] \min \left\{2^t| 2^t \geq e_1,e_2,\ldots,e_r \right\},
\end{equation}
where lcm means the least common multiple.
\end{theorem}

\begin{proof}
Please refer to page 150 of ``Algebraic Coding Theory"\cite{berlekamp2015algebraic}.
\end{proof}

\subsection{Comparison with former algorithms}
\begin{table}
	\caption{Comparison among three methods to calculate the period of the reversibility of LCA with $N$ cells}
	\hskip-0.0cm\
	\begin{tabular}{| p{2cm}| p{5cm}| p{2.2cm}| p{2.2cm}| p{2.2cm}|}
		
		\hline
		$r_L+r_R+1$ &   Linear Rule  & TMS \cite{2011_A_Martin_del_Rey} &  DFA \cite{yang2015reversibility} & PP[ours]\\
		\hline
		5 & 10011 & 0.032s & 0.004s & 0.0005s \\
		7  & 1000011 & 1.049s & 0.004s  & 0.0005s                          \\
		9  & 101100011 & 70.060s  & 0.004s  & 0.0005s                      \\
		11 & 10000001001  & 5187.539s  & 0.004s  & 0.0005s                  \\
		13 & 1000010011001  & Timeout  & 0.030s  & 0.0070s                 \\
		15 & 101100000000011  & Timeout  & 0.194s  & 0.0080s               \\
		17 & 10000000000101101  & Timeout  & 1.575s  & 0.0210s             \\
		19 & 1000000000010000001  & Timeout  & 16.422s  & 0.0660s          \\
		21 & 100000000000000001001   & Timeout & 153.955s  & 0.2570s       \\
		23 & 10000000000000000000011  & Timeout  & 1418.936s  & 0.6200s    \\
		25 & 1000000000000000000011011  & Timeout  & 11711.469s  & 5.1670s \\
		27 & 100000000000000000110000011  & Timeout  & Timeout  & 9.1210s \\
		\hline
		\multicolumn{2}{|c|}{Time Complexity(worst case)} & $O(2^{4(r_{L}+r_{R})})$ & $O(2^{r_{L}+2{r_{R}} - 1})$ & $O((r_{L}+r_{R})^{k})$\\
		\hline
	\end{tabular}
	\label{tab1}
\end{table}

Compared with the former result of \cite{2011_A_Martin_del_Rey} and \cite{yang2015reversibility},
we have greatly saved the time to calculate the period of the reversibility of LCA.
Our experiment by a laptop is showing in Table \ref{tab1}.
TMS is the method using transition matrix sequence\cite{2011_A_Martin_del_Rey}.
DFS is the method using DFA and de Bruijn graph\cite{yang2015reversibility}.
PP is our proposed method which calculates the period of the polynomial instead. This method successfully reduced the time complexity to $O((r_{L}+r_{R})^{k})$ where $k$ is a constant.
If the time is more than 4 hours, which is too long for us to wait, we label it as "timeout".

\section{Standard-basis-postfix(SBP) algorithm for verifying the reversibility of LCA}
\label{Reversibility}
The basic definitions of DFA has been proposed in \textbf{Section \ref{DFA}}. In this section we will propose an efficient algorithm, which is called "standard-basis-postfix algorithm",
to solve the other core problem - verifying reversibility in a period.

\subsection{Basic Description}

\begin{theorem}
	\label{extra}
	If there are two or more tuples in a node which have the same postfix, the node is irreversible, or else it is reversible.
\end{theorem}

\begin{proof}
	We need to mention one conclusion of the former DFA result\cite{yang2015reversibility},
	which is also mentioned as (2) of conclusion in Subsection \ref{Brief_Description_of_DFA}.
	It is:
	one node is reversible if and only if all tuples in it perform a complete permutation of $r_{R}$ binary elements.
	
	We can change this conclusion which has been proved in \cite{yang2015reversibility} to the theorem we need to prove: If there are two or more tuples in a node which have the same postfix, the node is irreversible, or else the node is reversible.
	
	So the theorem is proved.
\end{proof}

\begin{theorem}
	$postfix_{0} \equiv (0,...,0,...,0)$
\end{theorem}

\begin{proof}
We use \textbf{mathematical induction} to prove it:
	
\textbf{For the \bm{$1^{st}$} node},
it's obvious that $postfix_{0} = (0,...,0,...,0)$

\textbf{For the \bm{$k^{th}$} node}, we suppose $postfix_{0} = (0,...,0,...,0)$.

\textbf{For the \bm{$(k+1)^{th}$} node}, we just need to prove: the rightmost element of $postfix_{0}$ in the $(k+1)^{th}$ node is still $0$.
According to \textbf{Eq.(\ref{eq:second})},
	$s_{i}^{t+1} = [\lambda_{-r_{L}}s_{i-r_{L}}^{t} +...+ \lambda_{-1}s_{i-1}^{t} + \lambda_{0}s_{i}^{t} + \lambda_{1}s_{i+1}^{t} +...+ \lambda_{r_{R}}s_{i+r_{R}}^{t}]\quad \mod \quad 2$.
0-labelled edge means $s_{i}^{t+1} = 0$. Because $postfix_{0} = (0,...,0,...,0)$ in $k^{th}$ node, $s_{i-r_{L}}^{t}$,...,$s_{i}^{t}$,...,$s_{i+r_{R}-1}^{t}$ are all "0", so:
	$\lambda_{r_{R}}s_{i+r_{R}}^{t} = 0$
	
So the rightmost element $s_{i+r_{R}}$ in the $(k+1)^{th}$ node is 0.

So 	$postfix_{0} \equiv (0,...,0,...,0)$ and the theorem is proved.
\end{proof}

\begin{theorem}
	\label{thm6}
	Over ${\Bbb Z}_2$, ${\forall} r_{R}$, the initial node in a DFA with 0-labelled edge has $2^{r_R}$ tuples whose postfixes are  $2^{r_{R}}$ "0-1" permutation of length $r_R$.
	In the initial node, We can denote all $2^{r_R}-1$ postfixes except ${postfix}_0 \equiv 0$ with a linear combination of $r_R$ standard basis postfixes, which are as follows:
	\begin{center}
	${postfix}_{2^{0}} = (1,0,\ldots,0,\ldots,0) = {SBP}_{0}$\\
	${postfix}_{2^{1}} = (0,1,\ldots,0,\ldots,0) = {SBP}_{1}$\\
	...\\
	${postfix}_{2^{r_R-1}} = (0,0,\ldots,0,\ldots,1) = {SBP}_{r_R-1}$.
	\end{center}

\end{theorem}

\begin{proof}
	In the initial node,
	${\forall} {postfix}_{i}$ $(i \neq 0)$, we set a group of coefficients $(e_1,e_2,\ldots,e_i,\ldots,e_{r_R}) \neq (0,...0,...,0), e_{i} \in \{0,1\}$.
	Then the postfix can be denoted by the linear combination of $r_{R}$ standard basis postfixes mentioned in the theorem:
	\begin{center}
	${postfix}_i = e_{1}*{postfix}_{2^{0}} + e_2*{postfix}_{2^{1}} + \cdots + e_i * {postfix}_{2^{i}} + \cdots + e_{r_R}*{postfix}_{2^{r_R-1}}.$
	\end{center}
\end{proof}

\begin{theorem}
	The linear combination of standard basis postfixes with coefficients $e_{1}$, $e_{2}$...$e_{r_{R}}$ are used to represent $\quad{postfix}_{1}, ..., {postfix}_{i}, ..., {postfix}_{2^{r_{R}-1}}$.
	There are $2^{r_{R}}-1$ different types of linear combination with coefficients $e_{1}$, $e_{2}$...$e_{r_{R}}$  $(e_1,e_2,\ldots,e_i,\ldots,e_{r_R}) \neq (0,...0,...,0)$ corresponding to $2^{r_{R}}-1$ different types of $postfix_{i}$ ( $0<i<2^{r_{R}}$), which is a bijection.
\end{theorem}

\begin{proof}
	$e_{i} \in \{0,1\}$, so $e_{1}$, $e_{2}$...$e_{r_{R}}$ could be all "0-1" permutation of length $r_R$ except $(0,...,0,...,0)$, so there are $2^{r_{R}}-1$ different types of coefficients $e_{1}$, $e_{2}$...$e_{r_{R}}$ correpsonding to $2^{r_{R}}-1$ different types of postfixes.
	According to \textbf{Theorem \ref{thm6}}, 
	For $2^{r_{R}}-1$ different postfixes, there must be $2^{r_{R}}-1$ different coefficients for ${postfix}_i = e_{1}*{postfix}_{2^{0}} + e_2*{postfix}_{2^{1}} + \cdots + e_i * {postfix}_{2^{i}} + \cdots + e_{r_R}*{postfix}_{2^{r_R}-1}(0<i<2^{r_{R}}).$
	So there is a bijection between $2^{r_{R}}-1$ different kinds of coefficients $e_{1}$, $e_{2}$...$e_{r_{R}}$ and ${postfix}_{1}, ... ,{postfix}_{2^{r_{R}}-1}$.
\end{proof}

\begin{figure}[!htbp]
	\small
	\centering
	\includegraphics[width=18cm]{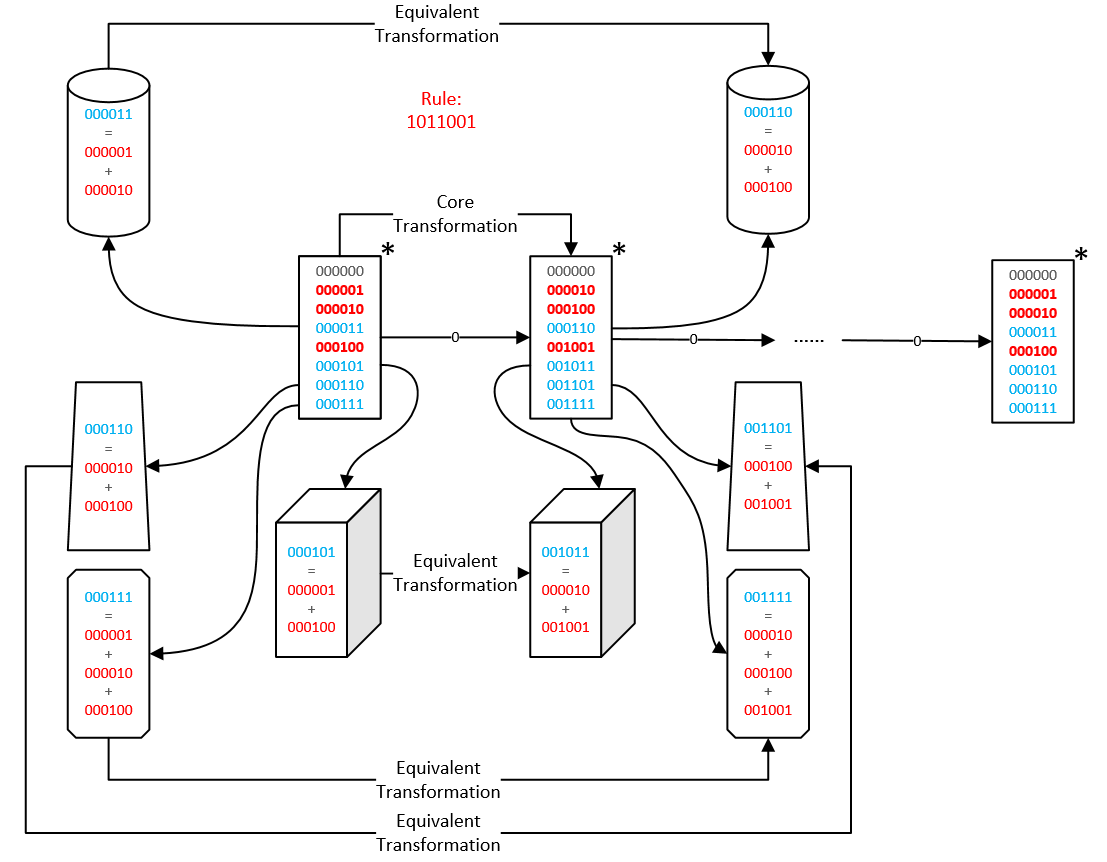}
	\caption{Use $r_R$ tuples to represent $2^{r_R}- r_{R} - 1$ tuples in the initial node, where $r_L=r_R=3$. Every tuple transformed from the initial node can still be represented by the linear combination of tuples in its current node of index $2^{j}(j = 0,1,..,r_{R}-1)$ with the same coefficients $e_{1}$, $e_{2}$...$e_{r_{R}}$ as it is represented in the initial node.
	}\label{flow_split}
\end{figure}

\textbf{Fig.\ref{flow_split}} shows how to use linear combination of $r_R$ tuples to represent all the other $2^{r_{R}} - r_{R} - 1$ tuples. The same linear combination coefficients retain in the following nodes.

\begin{theorem}
	\label{bijection}
	Every tuple transformed from the initial node can still be represented by the linear combination of tuples in its current node of index $2^{j}(j = 0,1,..,r_{R}-1)$ with the same coefficients $e_{1}$, $e_{2}$...$e_{r_{R}}$ as it is represented in the initial node, which means:\\
	\textbf{If in node $k$:}\\${tuple}_{i}$ =
	$\sum_{j = 0}^{r_{R}-1} {e_{j+1}} {tuple}_{2^{j}}$\\
	= $e_{1}{tuple}_{2^{0}} + e_{2}{tuple}_{2^{1}} +...+ e_{r_{R}}{tuple}_{2^{r_{R}-1}}$,\\
	\textbf{then in node $k+1$:}\\
	${tuple}_{i}$ =
	$\sum_{j = 0}^{r_{R}-1} {e_{j+1}} {tuple}_{2^{j}}$\\
	= $e_{1}{tuple}_{2^{0}} + e_{2}{tuple}_{2^{1}} +...+ e_{r_{R}}{tuple}_{2^{r_{R}-1}}$,
	\textbf{The coefficients \bm{$e_{1}, e_{2}, ..., e_{r_{R}}$} are the same between node $k$ and node $k+1$.}
	\label{retain}
\end{theorem}

\begin{proof}
We use \textbf{mathematical induction} to prove it:

\textbf{For the \bm{$1^{st}$} node} (the initial node):
Every tuple have the same prefix "0...0"(all-zeros) of length $r_{L}$, and we have proved every ${postfix}_{i}$ in the initial node can be represented by the linear combination of standard basis postfixes according to \textbf{Theorem \ref{thm6}}. So it's obvious\\ ${tuple}_{i} = ({prefix}_{i} , {postfix}_{i})$ = $e_{1}*({prefix}_{2^{0}} , postfix_{2^{0}}) + e_2*({prefix}_{2^{1}} , {postfix}_{2^{1}}) + \cdots + e_i * ({prefix}_{2^{i}} , {postfix}_{2^{i}}) $ + $\cdots + e_{r_R}*({prefix}_{2^{r_R-1}},{postfix}_{2^{r_R-1}}).$\\

\textbf{For the \bm{$k^{th}$} node}, we suppose the theorem is true.

\textbf{For the \bm{$({k+1})^{th}$} node}, we just need to prove: the rightmost element of its transformed tuple in the $(k+1)^{th}$ node can still be represented by the same group of tuples with index $2^{j}(j = 0,1,..,r_{R}-1)$ and the same coefficients $e_{1}$, $e_{2}$...$e_{r_{R}}$ as node $k^{th}$, and we prove it as follows:

\textbf{For any \bm{$tuple_{i}$} in node $k$}, we can denote ${tuple}_{i}$ as follows:
	${tuple}_{i}$ = 
	$\sum_{j = 0}^{r_{R}-1} {e_{j+1}}{tuple}_{2^{j}}$\\
	= $e_{1}{tuple}_{2^{0}} + e_{2}{tuple}_{2^{1}} +...+ e_{r_{R}}{tuple}_{2^{r_{R}-1}}$\\
	= $(\sum_{j = 0}^{r_{R}-1} e_{j+1} * {tuple}_{2^{j}}[0],\sum_{j = 0}^{r_{R}-1} e_{j+1} * {tuple}_{2^{j}}[1],...,\sum_{j = 0}^{r_{R}-1} e_{j+1} * {tuple}_{2^{j}}[r_{L}-1], \sum_{j = 0}^{r_{R}-1} e_{j+1} * {tuple}_{2^{j}}[r_{L}], \sum_{j = 0}^{r_{R}-1} e_{j+1} * {tuple}_{2^{j}}[r_L+1], ..., \sum_{j = 0}^{r_{R}-1} e_{j+1} * {tuple}_{2^{j}}[r_{L}+r_{R}-1])$.
	
Use \textbf{Eq.(\ref{eq:second})} with 0-labelled edge, for \bm{${tuple}_{i}$}, the rightmost element in the next node is:

		$\lambda_{r_{R}} * s_{i+r_{R}}$ 
		(of ${postfix}_{i}$) = 
		$\lambda_{-r_{L}} * \sum_{j = 0}^{r_{R}-1} e_{j+1} * {tuple}_{2^{j}}[0] + \lambda_{-r_{L}+1} * \sum_{j = 0}^{r_{R}-1} e_{j+1} * {tuple}_{2^{j}}[1] + ... + \lambda_{-1} * \sum_{j = 0}^{r_{R}-1} e_{j+1} * {tuple}_{2^{j}}[r_{L}-1] + 
		\lambda_{0} * \sum_{j = 0}^{r_{R}-1} e_{j+1} * {tuple}_{2^{j}}[r_{L}]+ \lambda_{1} * \sum_{j = 0}^{r_{R}-1} e_{j+1} * {tuple}_{2^{j}}[r_L+1] +... 
		\lambda_{r_{R}-1} * \sum_{j = 0}^{r_{R}-1} e_{j+1} * {tuple}_{2^{j}}[r_{L}+r_{R}-1]$\\

\textbf{For any \bm{$tuple_{2^{j}}$} in node $k$}, use \textbf{Eq.(\ref{eq:second})} with 0-labelled edge, the rightmost element in the next node is:
	
$\lambda_{r_{R}} * s_{i+r_{R}}$
(of ${tuple}_{2^{j}}$) =
$\lambda_{-r_{L}} * {tuple}_{2^{j}}[0] + \lambda_{-r_{L}+1} * {tuple}_{2^{j}}[1] + ... + \lambda_{-1} * {tuple}_{2^{j}}[r_{L}-1]$ + $\lambda_{0} * {tuple}_{2^{j}}[r_{L}]+ \lambda_{1} * {tuple}_{2^{j}}[r_L+1] +...  {tuple}_{2^{j}}[r_{L}+r_{R}-1]$\\

From the relationship between the right parts of the two equations of $tuple_{i}$ and $tuple_{2^{j}}$, we know:
$\lambda_{r_{R}} * s_{i+r_{R}}$ (of \bm{${tuple}_{i}$})= 
$\lambda_{r_{R}} * \sum_{j = 0}^{r_{R}-1} e_{j+1} * s_{i+r_{R}}$ (of \bm{${tuple}_{2^{j}}}$),
from which we could conclude that the rightmost number in node $k+1$ can also be represented by the same coefficients $e_{1}$, $e_{2}$...$e_{r_{R}}$ as node $k$.
Furthermore, according to (2) in "Important Theorem" in \textbf{Section \ref{Brief_Description_of_DFA}}: Two adjacent nodes have the same elements except the leftmost element of the first node and the rightmost element of the second node.

So the coefficients $e_{1}$, $e_{2}$...$e_{r_{R}}$ are the same between $k^{th}$ node and $({k+1})^{th}$ node.

So the theorem is proved.
	
\end{proof}

With the explanation of 0-labelled edge in \textbf{Section \ref{Brief_Description_of_DFA}}, we define $\stackrel{L}{\longrightarrow} \quad = \quad \stackrel{L_{1}}{\longrightarrow}\stackrel{L_{2}}{\longrightarrow}...\stackrel{L_{n}}{\longrightarrow}$ as a group of transformation where $\stackrel{L_{i}}{\longrightarrow}$ represent a 0-labelled transformation from a tuple in a previous node to its corresponding tuple in the next node.

\begin{example}
	Over ${\Bbb Z}_2$, we assume $a$, $b$ are two postfixes. If
	$a\stackrel{L}{\longrightarrow}a_1$, $b\stackrel{L}{\longrightarrow}b_1$,
	then $a + b\stackrel{L}{\longrightarrow}a_1 + b_1$.
	
	In \textbf{Fig.\ref{flow_split}},
	$r_R=3$, we have $2^{3}-1$ postfixes except ${postfix}_{0}$
	and these postfixes are $(0,0,1)$, $(0,1,0)$, $(0,1,1)$, $(1,0,0)$, $(1,0,1)$, $(1,1,0)$, $(1,1,1)$. 
	Then we can use the linear combination of $(0,0,1)$, $(0,1,0)$ and $(1,0,0)$ to denote all the other four postfixes as follows: 
	$(0,1,1)=(0,0,1) + (0,1,0)$, $(1,0,1)=(0,0,1) + (1,0,0)$, 
	$(1,1,0)=(0,1,0) + (1,0,0)$, $(1,1,1)= (0,0,1) + (0,1,0) + (1,0,0)$.
	
	We define the transformation which is on the top of \textbf{Fig.\ref{flow_split}} as $L_{1}$,\\
	$(0,1,1)=(0,0,1) + (0,1,0)$.
	If $(0,0,1)\stackrel{L_{1}}{\longrightarrow}(0,1,0)$, $(0,1,0)\stackrel{L_{1}}{\longrightarrow}(1,0,0)$,
	then $(0,1,1) = (0,0,1) + (0,1,0)\stackrel{L_{1}}{\longrightarrow} (0,1,0) + (1,0,0) = (1,1,0)$.
\end{example}

\begin{theorem}
\label{thm7}
${Trans}_{1}, {Trans}_{2},..., {Trans}_{r_{R}}$ are $r_R$ postfixes in a node which transformed from ${SBP}_1$, $\ldots$, ${SBP}_i$, $\ldots$, ${SBP}_{r_R}$ in the initial node. If ${Trans}_{1}, {Trans}_{2},..., {Trans}_{r_{R}}$ are linearly independent,
we classify this node as reversible, or otherwise we classify the node as irreversible.
\end{theorem}

\begin{proof}
For the former part of this theorem:
if ${Trans}_{1}, {Trans}_{2},..., {Trans}_{r_{R}}$ are linearly independent, Then according to \textbf{Theorem \ref{bijection}}, ${\forall}$ $(e_{1}$, $e_{2}$...$e_{r_{R}}) \neq (0,...,0,...,0)$, we have $e_1*{Trans}_1+\cdots+e_i*{Trans}_i+\cdots+e_{r_R}*{Trans}_{r_R} \neq (0,...,0,...,0)$ in this node. \\
Next we need to prove: for any two $postfix_{m}$ and $postfix_{n} (0 < m,n < r_{R})$, $e_1*{Trans}_1+\cdots+e_i*{Trans}_i+\cdots+e_{r_R}*{Trans}_{r_R}$ won't be the same.

We use reduction to absurdity to prove it:
First define $postfix_{m}$ and $postfix_{n} (0 < m,n < r_{R})$, $postfix_{m}$ = $e_1^{m}*{Trans}_1+\cdots+e_i^{m}*{Trans}_i+\cdots+e_{r_R}^{m}*{Trans}_{r_R}$
and $postfix_{n}$ = $e_1^{n}*{Trans}_1+\cdots+e_i^{n}*{Trans}_i+\cdots+e_{r_R}^{n}*{Trans}_{r_R}$.
We suppose $postfix_{m}$ = $postfix_{n}(0 < m,n < r_{R})$. Because $e_{i} \in \{0,1\}$, according to \textbf{Theorem \ref{bijection}}, 

there must be $postfix_{o} = e_1^{o}*{Trans}_1+\cdots+e_i^{o}*{Trans}_i+\cdots+e_{r_R}^{o}*{Trans}_{r_R} = (0,...,0,...,0)$ where $e_1^{o} = e_1^{m} + e_1^{n}$, ..., $e_i^{o} = e_i^{m} + e_i^{n}$, ...,$e_{r_{R}}^{o} = e_{r_{R}}^{m} + e_{r_{R}}^{n}$(over ${\Bbb Z}_2$), which contradict with the conclusion: $e_1*{Trans}_1+\cdots+e_i*{Trans}_i+\cdots+e_{r_R}*{Trans}_{r_R} \neq (0,...,0,...,0)$ in this node.

So every two postfixes in this node is different. According to \textbf{Theorem \ref{extra}}, the node is reversible. So the former part of the theorem is proved.

For the latter part of this theorem:	
if ${Trans}_{1}, {Trans}_{2},..., {Trans}_{r_{R}}$ has a linear correlation with coefficients ($e_1$,...,$e_i$,...,$e_{r_R}$) and 
$e1*{Trans}_1+\cdots+e_i*{Trans}_i+\cdots+e_{r_R}*{Trans}_{r_R}= (0,...,0,...,0)$ in the current node. 
Then according to \textbf{Theorem \ref{bijection}}, there must be an initial postfix ${postfix}_{i}$ in the first node which equals to
$e_1*{SBP}_1+\cdots+e_i*{SBP}_i+\cdots+e_{r_R}*{SBP}_{r_R}=0$, i.e.$(e_{1},e_{2},...,e_{r})$,
and it is transformed to $(0,\ldots,0,\ldots,0)$ in the current node.
So we have two or more postfixes equal to ${postfix}_0=(0,\cdots,0,\cdots,0)$ in current node including ${postfix}_0 \equiv 0$.
According to \textbf{Theorem \ref{extra}}, we classify the current node as irreversible.

\end{proof}

\begin{example}

When $r_R=5$, if we have one $postfix_{i}$ = $(1,0,1,0,1)$ in the initial node,\\
if $ (1,0,1,0,1)\stackrel{L}{\longrightarrow} (0,0,0,0,0) $, and
\begin{center}
$ (0,0,0,0,1)\stackrel{L}{\longrightarrow}a $\\
$ (0,0,0,1,0)\stackrel{L}{\longrightarrow}b $\\
$ (0,0,1,0,0)\stackrel{L}{\longrightarrow}c $\\
$ (0,1,0,0,0)\stackrel{L}{\longrightarrow}d $\\
$ (1,0,0,0,0)\stackrel{L}{\longrightarrow}e $\\
\end{center}
Based on \textbf{Theorem \ref{thm7}},
we can conclude the current node is reversible if $a, b, c, d, e$ are linearly independent.
On the other hand, if there is a linear correlation among $a, b, c, d, e$, we can conclude the current node is irreversible.
\end{example}

\subsection{Process summary}
\begin{figure}[!htbp]
\small
\centering
\includegraphics[width=18cm]{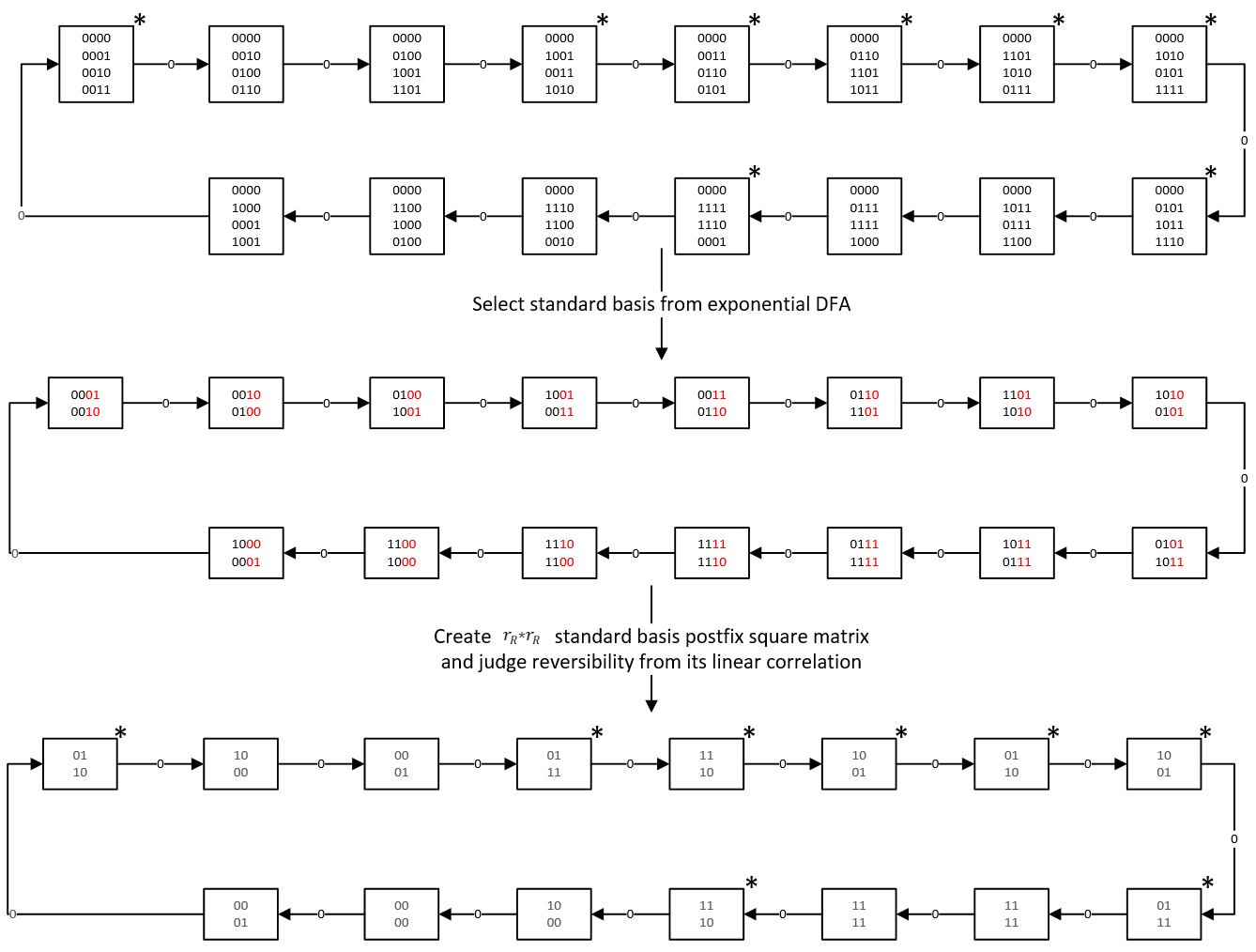}
\caption{Use standard basis postfix to judge reversibility}
\label{linear_subset_sequence}
\end{figure}

\textbf{Fig. \ref{linear_subset_sequence}} shows the process of using standard basis postfix to judge whether a node is reversible in a period.

\textbf{SBP(Standard-basis-postfix) algorithm }
\begin{description}
\item[Step1.] Take $r_R$ tuples which contain $r_R$ standard basis postfixes in initial node to construct a $r_R*(r_L + r_R)$ subset node.

\item[Step2.] Construct 0-labelled edge to the next node.

\item[Step3.] Repeat \textbf{Step2} until the node number reach the period we get in Section.\ref{Period} to form a 0-edged circle back to initial node. 

\item[Step4.] Construct a square matrix with size $r_R*r_R$ formed by standard basis postfixes.

\item[Step5.] Row reduce the matrix to an echelon form to determine whether this matrix is linearly independent.

\item[Step6.] Determine whether the current node is reversible by \textbf{Step5} and mark reversible node with "*".

\end{description}

\subsection{Comparison with former algorithm}
Finally we compare this algorithm with the former DFA in \cite{yang2015reversibility}.

In the original DFA algorithm, we have to store all $2^{r_R}$ tuples with length $r_L+r_R$ and verify if they repeat,
so we have a time complexity of $O\left(r_R*2^{r_R}\right)$ and a space complexity of $O\left((r_L+r_R)*2^{r_R}\right)$.
Thus for LCA of a slightly big $r_R$, we soon get memory overflow or it becomes incalculable for time.

In current standard-basis-postfix algorithm, we only need to have $r_R$ rows and verify the linear correlation
in this $r_R*r_R$ matrix by row reducing the matrix to an echelon form.
In this way, we reduce time complexity to $O(r_R^3)$ and space complexity to $O\left(r_R*(r_L+r_R)\right)$.

The comparison is shown in \textbf{Table \ref{tab2}}.\\

\begin{table}
	\caption{Comparison among three methods to verify the reversibility of LCA}
	\centering
\begin{tabular}{ccc}
	
	\hline
	Algorithm &  Time complexity    & Spatial complexity  \\
	\hline
	DFA   & $O\left(r_R*2^{r_R}\right)$   & $O\left((r_L+r_R)*2^{r_R}\right)$    \\
	Standard basis postfix 	 & $O({r_{R}}^3)$		& $O\left(r_R*(r_L+r_R)\right)$\\	\hline
\end{tabular}
	\label{tab2}
\end{table}

\section{Generate LCA rules with given period $T$ and give a lower bound of these rules' quantity}
\label{Generate}

\subsection{Give a unique period factorization}


Now, we consider an interesting inverse problem: Given a positive integer $T$,
can we generate an LCA rule whose period of the reversibility is $T$? 
If the answer is 'Yes', how to generate such LCA rule? And how many LCA rules with the period $T$ can we generate?

After analysis of \textbf{Theorem \ref{thm_fac}}, we divide this formula
into two parts, the former part is $lcm[p(f_1), p(f_2), \ldots, p(f_r)]$,
the latter part is $\min\{2^t|2^t \ge e_1, e_2,\ldots, e_r\}$.
We use $U$ to represent the former part and $V$ to represent the latter part from now on in this paper.

Suppose $T$ is the given number as a period. Our goal is to find an appropriate LCA rule whose polynomial is
$f(x)$ such that the period of the reversibility of it is $T$, i.e., $T=p(f)$.

First, let $T=UV$, where $U$ is an odd integer and $V=2^t$.

\begin{lemma}
	\label{lemma_odd}
	The period of any irreducible polynomial is odd over ${\Bbb Z}_2$.
\end{lemma}

\begin{proof}
	Over ${\Bbb Z}_2$, in the process of calculating irreducible polynomials' period in \textbf{Section \ref{cal_irre}},
	we know that $p(f)|2^n-1$. And it's obvious that $2^n-1$ must be an odd number.
	So for the \textbf{Theorem \ref{thm_fac}}, $p(f)$ must be an odd number.
\end{proof}

First according to \textbf{Lemma \ref{lemma_odd}}, we get the prime factorization of $p(f)$:
$p(f)=2^t\times m_1^{n_1} m_2^{n_2}\cdots m_k^{n_k}$, in which $m_i$ is an odd prime factor,
$V=2^t$ and $U=m_1^{n_1} m_2^{n_2}\cdots m_k^{n_k}$.

\begin{table}
	\begin{center}
		\caption{Part of period table of some irreducible polynomials over ${\Bbb Z}_2$}
		\begin{tabular}{| p{2cm}| p{2cm}| p{2cm}|}
			\hline
			deg($f(x)$) &  $f(x)$  & period\\
			\hline
			1& 11& 1\\
			2& 111& 3\\
			3& 1011& 7\\
			4& 11111& 5\\
			5& 100101& 31\\
			8& 100111001& 17 \\
			9& 1000000011& 73\\
			10& 11111111111& 11\\
			\hline
		\end{tabular}
		\label{tab3}
	\end{center}
\end{table}

Suppose for a given prime number $w$, there are $g(w)$ irreducible polynomials in the period table whose periods equal $w$.

Then for each prime factor in $U$ whose power is $1$, we find period $m_i$ in the polynomial period table and
directly give its corresponding polynomials $f_{i_1}, f_{i_2},..., f_{i_{g(m_i)}}$.
For those $m_{i}$ whose powers are bigger than $1$, we use the following \textbf{Theorem \ref{thm3}} 
and \textbf{Theorem \ref{thm_power}} to calculate them.

\begin{theorem}
	\label{thm3}
	Let $f(x)$ be an irreducible polynomial of period $n$ over ${\Bbb Z}_2$, and let t be a prime, $t \nmid 2$.
	If $t \mid n$, then every irreducible factor of $f(x^t)$ has period tn.
	If $t \nmid n$, then one irreducible factor of $f(x^t)$ has period $n$ and the other factors have period $tn$.
\end{theorem}

\begin{proof}
	The proof of this theorem can be found in page 153 of Algebraic Coding Theory\cite{berlekamp2015algebraic}.
\end{proof}

\begin{theorem}
	\label{thm_power}
	For any prime factors $m_i$ in $U$, we can get polynomials whose periods equal to any powers of $m_i$.
\end{theorem}

\begin{proof}
	From \textbf{Lemma \ref{lemma_odd}}, we know that the period of irreducible polynomial $p(f_i)$ over ${\Bbb Z}_2$ must be an odd number.
	
	So according to \textbf{Theorem \ref{thm3}}, any $m_i$ in $U$ meet the requirement $m_i \nmid 2$($t \nmid 2$).
	Then we suppose polynomials $h_1(x)$ has periods $m_i$.
	We can get an (maybe more than one) irreducible polynomial as a factor of $h_1(x^{m_i})$
	whose period equals to $m_i*m_i = m_i^2$ by using \textbf{Theorem \ref{thm3}}. We denote it as $h_2(x)$.
	
	Similarly we can get an (maybe more than one) irreducible polynomial
	as a factor of $h_2(x^{m_i})$ whose period equals to $m_i^2*m_i=m_i^3$ by using \textbf{Theorem \ref{thm3}}.
	We denote it as $h_3(x)$.
	
	By using this methods recursively, we can get polynomials whose periods equal to any powers of $m_i$.
\end{proof}

\begin{theorem}
	\label{thm_2_power}
	For any integer $t\ge 0$, we can get a polynomial whose period equals to $2^t$.
\end{theorem}

\begin{proof}
	We consider the polynomial $g(x)=(x+1)^s$, where $2^{t-1}+1\le s\le 2^t$. 
	Because the period of the irreducible polynomial $x+1$ equals to $1$, it is easy to see that $g(x)$
	has period equals to $2^t$ according to \textbf{Theorem \ref{thm_fac}}.
\end{proof}


Now we could find a group of polynomials which meet the requirement
$U = lcm[p(f_{1}), p(f_{2}), \ldots, p(f_{r})]$ from a period table which maps period to irreducible polynomials over ${\Bbb Z}_2$.(We draw part of the period table as \textbf{Table \ref{tab3}}. As for the full table, please refer to \cite{marsh1957table}).
Furthermore, We need to give a group of powers $\{e_1, e_2,\ldots, e_r\}$ to control the part $V$ to meet the requirement $2^{t-1} < \max\{e_1, e_2,\ldots, e_r\} \leq 2^t$ too.

\subsection{The lower bound for the number of LCA rules with a given period}
Suppose the given period $T$ has the prime factorization

\begin{equation}
T= m_1^{e_1}m_2^{e_2} \cdots m_r^{e_r} * 2^t = UV,
\label{thm_T_fac}
\end{equation}
where $m_1,m_2\cdots,m_r$ are odd primes.
We have supposed there are $g(m_i)$ irreducible polynomials whose periods equal to $m_i$.

\begin{theorem}
	There are at least $(2^{tr}-2^{(t-1)r})\prod_{i=1}^r{g(m_i)}$ different kinds of
	polynomials for given period $T$ without irreducible polynomial x+1 as a factor of $f(x)$. 
	\label{thm_type_exclude}
\end{theorem}

\begin{proof}
	From \textbf{Theorem \ref{thm_fac}}, we can construct the polynomial of the following form:
	\[ f(x)=f_1^{e_1}(x)\cdots f_r^{e_r}(x),  \]
	where the periods of $f_1(x), \ldots, f_r(x)$ are $m_1, \ldots, m_r$ respectively.
	Notice that $e_{1}, e_{2},..., e_{r}$ must meet requirements:
	$2^{t-1} < \max\{e_{1}, e_{2},..., e_{r}\} \leq 2^{t} $.
	So we have $2^{tr} - 2^{(t-1)r}$ kinds of $e_1, e_2,\ldots, e_r$.
	Then for each prime factor $m_i$, if we have $g(m_i)$ kinds of corresponding irreducible polynomials,
	there must be at least $(2^{tr}-2^{(t-1)r})\prod_{i=1}^r{g(m_i)}$ 
	different kinds of combinations to get the polynomials for the given period $T$ 
	without irreducible polynomial x+1 as a factor of $f(x)$. 
\end{proof}

Due to the period of $f_0(x)=x+1$ is $1$, it is optional as a factor of $f(x)$ when we generate the polynomial with given period.
If we take $x+1$ into consideration, and set $f_0=x+1$, we will get the next theorem.

\begin{theorem}
	There are at least $(2^t-1)(2^{tr} - 2^{(t-1)r})\prod_{i=1}^{r}{g(m_i)}+2^{tr}\prod_{i=1}^{r}{g(m_i)}$ extra different kinds of
	polynomials for period $T$ with $x+1$ as a factor of $f(x)$.
\end{theorem}

\begin{proof}
	Consider the polynomial of the following form: 
	\[ f(x)=f_0^{e_0}(x) f_1^{e_1}(x)\cdots f_r^{e_r}(x).  \]
	If we demand that $f(x)$ has the factor $f_0(x)=x+1$, then $e_0=1,2,\ldots,2^t$.
	And when $e_0=2^t$, the requirement that $\max\{e_{1}, e_{2},..., e_{r}\} > 2^{t-1}$ is no need,
	so we have at least 
	\[ (2^t-1)(2^{tr} - 2^{(t-1)r})\prod_{i=1}^{r}{g(m_i)}+2^{tr}\prod_{i=1}^{r}{g(m_i)} \]
	extra polynomials for period $T$ with $x+1$ as a factor of $f(x)$.
\end{proof}


Finally, We translate all $f(x)$ back to LCA rules whose periods equal to the given one.

\begin{example}
	For example, if we have a period $T=84=2^2*3*7$, then according to Eq.(\ref{thm_T_fac}), $V=2^2$, $U=3*7$,
	and we know $g(3)\ge 1$, $g(7) \ge 1$ by looking up \textbf{Table \ref{tab3}},
	
	For period $3$, $f_1(x)=x^2+x+1$. For period $7$, $f_2(x)=x^3+x+1$.
	Therefore, we can get at least $(4 * 4 - 2 * 2) * 1 * 1 $ polynomials without $x+1$ as a factor.
	
	If we take $f_0(x) = x + 1$ into consideration, we can get $3 * (4 * 4 - 2 * 2)* 1 * 1 + 4 * 4* 1 * 1 $ more polynomials.
	And we have at least $(1+3) (4 * 4 - 2 * 2)* 1 * 1 + 4 * 4* 1 * 1 =64$ polynomials in total.
	
	Finally we translate these polynomials back to LCA rules.
\end{example}

\begin{example}
	Another example, if we have a period $T=360=2^3*3^2*5$, then according to \textbf{Lemma \ref{lemma_odd}},
	$V=2^{3}$, $U=3^{2}*5$, and we know $g(3)\ge 1$, $g(5)\ge 1$ by looking up \textbf{Table \ref{tab3}}.
	
	For period $3^{2}$, we take two steps:
	\[ pre\_f_{1}(x) = x^{2} + x + 1,
	mid\_f_{1}(x) = pre\_f_{1}(x^{3}) = (x^{3})^{2} + (x^{3}) + 1 = x^{6} + x^{3} + 1. \]
	If $mid\_f_{1}(x)$ is irreducible, then it equals to the final $f_1(x)$,
	else, from \textbf{Theorem \ref{thm3}}, we know that it has an irreducible factor whose period is $3^2$, 
	and we take this factor as $f_1(x)$.
	
	For period $5$, $f_2(x)=x^4+x^3+x^2+x+1$.
	
	Therefore, we can get at least $ (8 * 8 - 4 * 4) * 1 * 1 $ polynomials without $x+1$ as a factor.
	
	If we take $f_0(x)=x+1$ into consideration, we can get $7 * (8 * 8 - 4 * 4) * 1 * 1 + 8 * 8 * 1 * 1$ more polynomials.
	Then we have at least $(1+7) * (8 * 8 - 4 * 4) * 1 * 1 + 8 * 8 * 1 * 1=448$ polynomials in total.
	
	Finally we translate these polynomials back to LCA rules.
\end{example}

\subsection{Algorithm to generate the polynomial with given period $T$}
\begin{description}
	\item[Step1.] For a given period $T$, factorize it to prime factors
	\[T=2^t\times m_1^{e_1}m_2^{e_2}\cdots m_r^{e_r},\]
	where each $m_i(1\le i\le r)$ is an odd prime.
	
	\item[Step2.] Set $V=2^t$ and $U=m_1^{e_1}m_2^{e_2}\cdots m_r^{e_r}$.
	
	\item[Step3.]
	For each $e_i=1$, look up the corresponding polynomial $f_i(x)$ with period $m_i$ from the polynomial period table over ${\Bbb Z}_2$.
	
	For each $e_j>1$, look up the corresponding polynomial $f_{j_1}(x)$ with period $m_j$ from the polynomial period table,
	and find at least one irreducible polynomial with period $m_i^2$ as a factor of $f_{j_1}(x^{m_j})$, denote it as $f_{j_2}(x)$. 
	Similarly we can find least one irreducible polynomial with period $m_i^3$ as a factor of $f_{j_2}(x^{m_j})$, 
	denote it as $f_{j_3}(x)$. From \textbf{Theorem \ref{thm3}}, continually doing this, 
	we could always get an irreducible polynomial whose period equals to $m_i^{e_i}$ at last.
	
	\item[Step4.]
	Get $f(x)$ by combining these $f_{i}(x)$ using \textbf{Theorem\ref{thm_fac}}.
	
	\item[Step5.]
	Translate these $f(x)$ back to LCA rules. The lower bound of the numbers of LCA rules, 
	whose period is $T$, 
	is $2^t(2^{tr} - 2^{(t-1)r})\prod_{i=1}^{r}{g(m_i)}+2^{tr}\prod_{i=1}^{r}{g(m_i)}$.
\end{description}

\section{Conclusions\label{Conclusion}}
The paper has proposed polynomial substitution and standard basis postfix algorithms to efficiently calculate reversibility of LCA, which reduce former algorithms' exponential complexity to polynomial complexity. Furthermore, a novel perspective is proposed to generate rules from period conversely.

The polynomial algorithm replace the former DFA to calculate the period, which reduce time complexity from $O(2^{r_{L}+2{r_R}-1})$ to $O((r_L + r_R)^k)$, where k is a constant. With this period, we have a range for further reversibility verification.

The standard-basis-postfix algorithm verify reversibility for a specific node, 
which decreases the time complexity from O($2^{r_{R}}$) to O(${r_{R}}^{3}$), space complexity from O($(r_{R}+r_{L}) * 2^{r_{R}}$) to O($(r_{R}+r_{L}) * r_{R}$).

Moreover, when we are given a positive integer as the period of the reversibility of a specific LCA, 
based on the polynomial formula mentioned before, we can conversely give LCA rules corresponding with it.

To summarize, the main objective of this paper is to propose a more efficient way to cope with reversibility of LCA with large size, 
which is incalculable before. With this improvement, reversible LCA will have better and broader applications in data encryption, 
decryption and error-correcting codes.

Additional further work may include the analysis of properties of all kinds of LCA with various boundary conditions and extend the result of this article to multidimensional cases.






\begin{thebibliography}{AA}
\bibitem{2012_H_Akin}
H. Akin, F. Sah and I. Siap,
On 1D Reversible Cellular Automata with Reflective Boundary Over the Prime Field of Order $p$,
Int. J. Mod. Phys. C, 23(1)(2012) 1250004:1-13.

\bibitem{berlekamp1967factoring}
E. Berlekamp,
Factoring polynomials over finite fields,
Bell System Technical Journal, 46(8)(1967) 1853-1859.
\bibitem{bakhshandeh2013authenticated}
A. Bakhshandeh and Z. Eslami,
An authenticated image encryption scheme based on chaotic maps and memory cellular automata,
Optics and Lasers in Engineering, 51(6)(2013) 665-673.
\bibitem{berlekamp2015algebraic}
E. Berlekamp, Algebraic coding theory, World Scientific, New Jersey, 2015.

\bibitem{cantor1981new}
D. G. Cantor and H. Zassenhaus, A new algorithm for factoring polynomials over finite fields,
Mathematics of Computation, (1981) 587-592.
\bibitem{chang2017reversibility_1}
C. Chang, J. Su, etc.,
Reversibility of linear cellular automata on Cayley trees with periodic boundary condition,
Taiwanese Journal of Mathematics, 21(6)(2017) 1335-1353.
\bibitem{chang2017reversibility_2}
C. Chang, J. Su, H. Ak{\i}n and F. {\c{S}}ah,
Reversibility Problem of Multidimensional Finite Cellular Automata,
Journal of Statistical Physics, 168(1)(2017) 208-231.
\bibitem{2011_Z_Cinkir}
Z. Cinkir, A. Hasan and I. Siap,
Reversibility of 1D Cellular Automata with Periodic Boundary over Finite Fields $Z(p)$,
J. Stat. Phys., 143(4)(2011) 807-823.

\bibitem{das2010parallel}
D. Das and A. Ray,
A parallel encryption algorithm for block ciphers based on reversible programmable cellular automata,
arXiv:1006.2822, 2010.
\bibitem{2006_A_M_Del_Rey}
A. M. del Rey and G. Rodr\'iguez S\'anchez,
On the reversibility of $150$ Wolfram cellular automata, Int. J. Mod. Phys. C, 17(7)(2006) 975-983.
\bibitem{2009_A_Martin_Del_Rey}
A. M. del Rey and G. Rodr\'iguez S\'anchez,
Reversibility of a Symmetric Linear Cellular Automata, Int. J. Mod. Phys. C, 20(7)(2009) 1081-1086.
\bibitem{2011_A_Martin_del_Rey}
A. M. del Rey and G. Rodr\'iguez S\'anchez,
Reversibility of linear cellular automata, Appl. Math. Comput., 217(21)(2011) 8360-8366.
\bibitem{2013_A_Martin_del_Rey_b}
A. M. del Rey,
A Note on the Reversibility Of Elementary Cellular Automaton 150 With Periodic Boundary Conditions,
Rom. J. Inf. Sci. Tech., 16(4)(2013) 365-372.
\bibitem{2013_A_Martin_del_Rey_a}
A. M. del Rey and G. Rodr\'iguez S\'anchez,
On the invertible cellular automata $150$ over $F_p$, Appl. Math. Comput., 219(10)(2013) 5427-5432.

\bibitem{2007_L_Hernandez_Encinas}
L. Hern\'andez Encinas and A. Mart\'in del Rey,
Inverse rules of ECA with rule number $150$,
Appl. Math. Comput., 189(2)(2007) 1782-1786.

\bibitem{1990_J_Kari}
J. Kari, Reversibility of 2D cellular automata is undecidable, Physica D, 45(1-3)(1990) 379-385.
\bibitem{kari1992cryptosystems}
J. Kari, Cryptosystems based on reversible cellular automata, Manuscript, 1992.
\bibitem{kippenberger2013modeling}
S. Kippenberger, A. Bernd, D. Tha{\c{c}}i, etc.,
Modeling pattern formation in skin diseases by a cellular automaton,
The Journal of investigative dermatology, 133(2)(2013) 567.

\bibitem{maranon2005new}
G. {\'A}. Mara{\~n}{\'o}n, L. H. Encinas and A. M. del Rey,
A new secret sharing scheme for images based on additive 2-dimensional cellular automata,
Iberian Conference on Pattern Recognition and Image Analysis, 2005, 411-418.
\bibitem{marsh1957table}
R. W. Marsh,
Table of irreducible polynomials over GF(2) through degree $19$,
Office of Technical Services, US Department of Commerce, 1957.

\bibitem{2004_A_Nobe}
A. Nobe and F. Yura, On reversibility of cellular automata with periodic boundary conditions,
J. Phys. A-Math. Gen., 37(22)(2004) 5789-5804.

\bibitem{rickert1996two}
M. Rickert, K. Nagel, M. Schreckenberg and A. Latour,
Two lane traffic simulations using cellular automata,
Physica A: Statistical Mechanics and its Applications, 231(4)(1996) 534-550.

\bibitem{2012csah}
F. Sah, I. Siap and H. Akin,
Characterization of Three Dimensional Cellular Automata over $Z(m)$,
AIP Conference Proceedings, 1470, Aug.2012, 138-141.
\bibitem{1998_P_Sarkar}
Palash Sarkar and Rana Barua,
The set of reversible $90/150$ cellular automata is regular,
Discret. Appl. Math., 84(1-3)(1998) 199-213.
\bibitem{2012_J_C_Seck-Tuoh-Mora}
J. C. Seck-Tuoh-Mora, G. J. Martinez, R. Alonso-Sanz and N. Hernandez-Romero,
Invertible behavior in elementary cellular automata with memory,
Inf. Sci., 199(1)(2012) 125-132.
\bibitem{seck2017welch}
J. C. Seck-Tuoh-Mora, J. Medina-Marin, N. Hernandez-Romero, etc.,
Invertible behavior in elementary cellular automata with memory, Inf. Sci., 382(1)(2017) 81-95.
\bibitem{seck2018graphs}
J. C. Seck-Tuoh-Mora and G. J. Mart{\'\i}nez,
Graphs Related to Reversibility and Complexity in Cellular Automata,
Cellular Automata: A Volume in the Encyclopedia of Complexity and Systems Science, Second Edition, Springer, 2018, 479-492.
\bibitem{shoup1990new}
V. Shoup, New algorithms for finding irreducible polynomials over finite fields, Mathematics of Computation, 54(189)(1990) 435-447.
\bibitem{2010_I_Siap}
I. Siap, H. Akin and F. Sah,
Garden of eden configurations for 2-D cellular automata with rule $2460$ N,
Inf. Sci., 180(18)(2010) 3562-3571.
\bibitem{2011_I_Siap}
I. Siap, H. Akin and S. U\u{g}uz,
Structure and reversibility of 2D hexagonal cellular automata,
Comput. Math. Appl., 62(11)(2011) 4161-4169.
\bibitem{2012_I_Siap}
I. Siap, H. Akin and M. E. Koroglu,
Reversible Cellular Automata with Penta-Cyclic Rule and ECCs,
Int. J. Mod. Phys. C, 23(10)(2012) 50066:1-13.
\bibitem{souyah2016fast}
A. Souyah and K. M. Faraoun,
Fast and efficient randomized encryption scheme for digital images based on quadtree decomposition and reversible memory cellular automata,
Nonlinear Dynamics, 84(2)(2016) 715-732.
\bibitem{2000_K_Sutner}
K. Sutner, sigma-automata and Chebyshev-polynomials, Theor. Comput. Sci., 230(1-2)(2000) 49-73.

\bibitem{takesue1989ergodic}
S. Takesue,
Ergodic properties and thermodynamic behavior of elementary reversible cellular automata. I. Basic properties,
Journal of Statistical Physics, 56(3-4)(1989) 371-402.

\bibitem{2013_S_Uguz}
S. U\u{g}uz, H. Akin and I. Siap,
Reversibility Algorithms for 3-State Hexagonal Cellular Automata with Periodic Boundaries,
Int. J. Bifurcat. Chaos, 23(6)(2013) 1350101: 1-15.

\bibitem{von2001factoring}
J. Von Zur Gathen and D. Panario,
Factoring polynomials over finite fields: A survey,
Journal of Symbolic Computation, 31(1-2)(2001) 3-17.

\bibitem{wang2013novel}
X. Wang and D. Luan,
A novel image encryption algorithm using chaos and reversible cellular automata,
Communications in Nonlinear Science and Numerical Simulation, 18(11)(2013) 3075-3085.
\bibitem{wolfram2018cellular}
S. Wolfram, Cellular automata and complexity: collected papers, CRC Press, 2018.

\bibitem{2013_M_Yamagishi}
M. Yamagishi, Elliptic curves over finite fields and reversibility of additive cellular automata on square grids,
Finite Fields Th. App., 19(1)(2013) 105-119.
\bibitem{yang2015reversibility}
B. Yang, C. Wang and A. Xiang,
Reversibility of general 1D linear cellular automata over the binary field Z2 under null boundary conditions,
Information Sciences, 324(1)(2015) 23-31.

\bibitem{Zhang2002}
C. Zhang, Q. Peng and Y. Li, Encryption based on reversible cellular automata,
IEEE 2002 International Conference on Communications, Circuits and Systems and West Sino Expositions, 2(2002), 1223-1226.

\end{thebibliography}

\end{document}